\documentclass[conference]{IEEEtran}
\IEEEoverridecommandlockouts
\ifCLASSINFOpdf
\else
\fi
\usepackage{amsmath}
\usepackage{amssymb}
\usepackage{amsthm}
\usepackage{mathtools}
\usepackage{color}
\usepackage{wrapfig}
\usepackage{array}
\usepackage{scrextend}
\usepackage{multirow}
\usepackage{xcolor}
\usepackage{tikz}
\usetikzlibrary{shapes,arrows,positioning,decorations.pathreplacing,calc}

\definecolor{green1}{RGB}{81,107,62}
\definecolor{blue1}{RGB}{0,32,96}
\tikzstyle{block} = [rectangle, draw, text centered, rounded corners, minimum height=2em]

\newtheorem{thm}{Theorem}

\newtheorem{defn}{Definition}
\newtheorem{lemma}{Lemma}

\newtheorem{construction}{Construction}

\newcommand{\ve}[1]{\ensuremath{\boldsymbol{#1}}}

\newcolumntype{C}[1]{>{\centering\let\newline\\\arraybackslash\hspace{0pt}}m{#1}}

\newcommand{\sub}{\ensuremath{\mathbb{S}}}

\usepackage{color}

\begin{document}
\bstctlcite{IEEEexample:BSTcontrol}
\title{{\textbf{Anchor-Based Correction of Substitutions\\ in Indexed Sets}}\vspace{-1ex}}


%
\author{
	\IEEEauthorblockN{
		\textbf{Andreas Lenz}\IEEEauthorrefmark{1},
		\textbf{Paul H. Siegel}\IEEEauthorrefmark{2},
		\textbf{Antonia Wachter-Zeh}\IEEEauthorrefmark{1}, and
		\textbf{Eitan Yaakobi}\IEEEauthorrefmark{3}
	}
	
	\IEEEauthorblockA{
		\IEEEauthorrefmark{1}Institute for Communications Engineering, Technical University of Munich, Germany
	}
	\IEEEauthorblockA{
		\IEEEauthorrefmark{2}Department of Electrical and Computer Engineering, CMRR, University of California, San Diego, California
	}
	\IEEEauthorblockA{
		\IEEEauthorrefmark{3}Computer Science Department, Technion -- Israel Institute of Technology, Haifa, Israel
	}
	\textbf{Emails}: andreas.lenz@mytum.de, psiegel@ucsd.edu, antonia.wachter-zeh@tum.de, yaakobi@cs.technion.ac.il\vspace{-3ex}
	
	\thanks{This work was done in part while A. Lenz was visiting the computer science faculty of Technion -- Israel Institute of Technology, Israel. This work was supported by the Institute for Advanced Study (IAS), Technische Universit\"{a}t M\"{u}nchen (TUM), with funds from the German Excellence Initiative and the European Union's Seventh Framework Program (FP7) under grant agreement no.~291763. This work was also supported by NSF Grant CCF-BSF-1619053 and by the United States-Israel BSF grant 2015816.}
	
}


\maketitle

\begin{abstract}
	Motivated by DNA-based data storage, we investigate a system where digital information is stored in an unordered set of several vectors over a finite alphabet. Each vector begins with a unique index that represents its position in the whole data set and does not contain data. This paper deals with the design of error-correcting codes for such indexed sets in the presence of substitution errors. We propose a construction that efficiently deals with the challenges that arise when designing codes for unordered sets. Using a novel mechanism, called \emph{anchoring}, we show that it is possible to combat the ordering loss of sequences with only a small amount of redundancy, which allows to use standard coding techniques, such as tensor-product codes to correct errors within the sequences. We finally derive upper and lower bounds on the achievable redundancy of codes within the considered channel model and verify that our construction yields a redundancy that is close to the best possible achievable one. Our results surprisingly indicate that it requires less redundancy to correct errors in the indices than in the data part of vectors.
\end{abstract}


%
\IEEEpeerreviewmaketitle

\section{Introduction}
\begin{figure*}
	\newcommand{\xwidth}{0cm}
	\newcommand{\indexwidth}{1.5cm}
	\newcommand{\datawidth}{2.84cm}
	\newcommand{\mht}{0.6cm}
	\tikzstyle{block2} = [rectangle, draw, minimum height=\mht]
	\begin{tikzpicture}

	\node[block2 ,minimum width=\indexwidth] (index) {$1$};
	\node[left= 0 of index,minimum width=\xwidth] {$\ve{x}_1$};
	\node[block2, right= 0pt of index, minimum width=\datawidth] (data) {$\ve{u}_1$};
	
	\node[block2, below= .1cm of index,minimum width=\indexwidth] (index2) {$2$};
	\node[left= 0 of index2,minimum width=\xwidth] {$\ve{x}_2$};
	\node[block2, right= 0pt of index2, minimum width=\datawidth] (data2) {$\ve{u}_2$};

	\node[below = 0pt of index2,minimum height=\mht] (dots) {$\vdots$};

	\node[block2, below= 0pt of dots,minimum width=\indexwidth] (indexM) {$M$};
	\node[left= 0 of indexM,minimum width=\xwidth] {$\ve{x}_M$};
	\node[block2, right= 0pt of indexM, minimum width=\datawidth] (dataM) {$\ve{u}_M$};

	\node at (0,.65) {Index};
	\node at (\indexwidth/2+\datawidth/2,.65) {Data};

	\fill[gray, opacity=0.2, rounded corners] (-\indexwidth/2, .95cm) rectangle ($(indexM) + (0.5*\indexwidth, -0.55*\mht)$);
	
	\fill[lightgray, opacity=0.2, rounded corners] ($(data) + (-\datawidth/2, .95cm)$) rectangle ($(dataM) + (0.5*\datawidth, -0.55*\mht)$);

	\node[right= 1.5cm of data, minimum width=\xwidth] (Bx) {$\ve{x}_1'$};
	\node[block2 , right= 0cm of Bx, minimum width=\indexwidth] (Bindex) {$4$};
	\node[block2, right= 0pt of Bindex, minimum width=\datawidth] (Bdata) {$\ve{u}_1'$};
	
	\node[block2, below= .1cm of Bindex,minimum width=\indexwidth] (Bindex2) {$2$};
	\node[left= 0 of Bindex2,minimum width=\xwidth] (Bx2) {$\ve{x}_2$};
	\node[block2, right= 0pt of Bindex2, minimum width=\datawidth] (Bdata2) {$\ve{u}_2$};

	\node[below = 0pt of Bindex2,minimum height=\mht] (Bdots) {$\vdots$};

	\node[block2, below= 0pt of Bdots,minimum width=\indexwidth] (BindexM) {$7$};
	\node[left= 0 of BindexM,minimum width=\xwidth] (BxM) {$\ve{x}_M'$};
	\node[block2, right= 0pt of BindexM, minimum width=\datawidth] (BdataM) {$\ve{u}_M'$};
	
	\draw[->] ($(data2.east) + (0.25,-0.25)$) -- node[above]  {Perturb} ($(data2.east) + (1.35,-0.25)$);

	\node[right= 1.5cm of Bdata, minimum width=\xwidth] (Cx) {$\ve{x}_2$};
	\node[block2, right= 0cm of Cx, minimum width=\indexwidth] (Cindex) {$2$};
	\node[block2, right= 0pt of Cindex, minimum width=\datawidth] (Cdata) {$\ve{u}_2$};
	
	\node[block2, below= .1cm of Cindex,minimum width=\indexwidth] (Cindex2) {$1$};
	\node[left= 0 of Cindex2,minimum width=\xwidth] (Cx2) {$\ve{x}_5'$};
	\node[block2, right= 0pt of Cindex2, minimum width=\datawidth] (Cdata2) {$\ve{u}_5'$};

	\node[below = 0pt of Cindex2,minimum height=\mht] (Cdots) {$\vdots$};

	\node[block2, below= 0pt of Cdots,minimum width=\indexwidth] (CindexM) {$7$};
	\node[left= 0 of CindexM,minimum width=\xwidth] (CxM) {$\ve{x}_M'$};
	\node[block2, right= 0pt of CindexM, minimum width=\datawidth] (CdataM) {$\ve{u}_M'$};
	
	\draw[->] ($(Bdata2.east) + (0.25,0)$) to[out=0,in=180] ($(Cx.west) + (-0.15,0)$);
	\draw[->] ($(Bdata.east) + (0.25,0)$) to[out=0,in=180] ($(Cdots.west) + (-1.4,0)$);
	\draw[->] ($(Bdots.east) + (3.69,0)$) to[out=0,in=180] ($(Cx2.west) + (-0.15,0)$);
	\draw[->] ($(BdataM.east) + (0.25,0)$) to[out=0,in=180] ($(CxM.west) + (-0.03,0)$);
	
	\node at (11cm,.65) {Permute};
	
	\end{tikzpicture}
	\caption{Channel model for information storage in indexed sets. First, some sequences $\ve{x}_i$ are perturbed by substitution errors, resulting in $\ve{x}_i' = \ve{x}_i + \ve{e}_i$. Afterwards, the sequences can be permuted arbitrarily and hence all inherent information about their ordering is lost. Since the indices can be erroneous, too, it is not necessarily directly possible to reconstruct their ordering.}
	\label{fig:channel}
	\vspace{-.4cm}
\end{figure*}
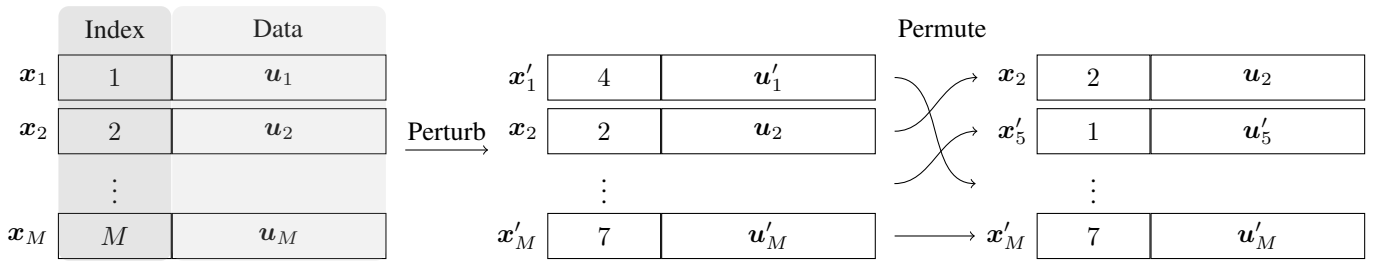
Consider a system where digital information is embodied in an unordered set of vectors and each vector holds a share of the whole data set. To combat the unordered nature of the data storage, such systems almost exclusively rely on indices, which are prepended to each vector and denote the position of that vector in the data set. An important example for a modern communication system of such a type is internet routing, where data is split into packets and transmitted over a network. Since packets can have different propagation times over the network, they might be received in a different order and hence the ordering of the packets is lost. Another important example for such a system is deoxyribonucleic acid (DNA)-based data storage, which is the main focus of this paper.

Data storage in DNA is a novel technology that, due to recent advancements in biochemical mechanisms of synthesizing and sequencing DNA molecules, has advanced to be a highly competitive candidate for long-term archival storage of digital data. This is since DNA-based storage has several important features that stand out with respect to conventional digital data storage systems, such as tapes and hard disk drives. These include outstandingly high data densities and long-term robustness. Due to its chemical structure, from a coding theoretic point of view, DNA can be seen as a vector over symbols \{A,C,G,T\}, which abbreviate the four nucleotides adenine (A), cytosine (C), guanine (G), and thymine (T). Single DNA strands can be synthesized chemically and modern DNA synthesizers can concatenate the four DNA nucleotides to form almost any possible sequence. This process enables the storage of digital data in DNA. The data can be read back with common DNA sequencers, while the most popular ones use DNA polymerase enzymes.

Using DNA as a storage medium for digital data was envisioned by Feynman in his famous speech ``There's plenty of room at the bottom" and also by Baum \cite{Bau95}. It took several decades until first experiments of Church et al.~\cite{CGK12} and later Goldman et al. ~\cite{GBCDLSB13} demonstrated the viability of \emph{in vitro} DNA storage on a large scale. In the next years, many experiments followed, including Grass et al. ~\cite{GHPPS15} who successfully employed error-correcting codes to recover the data. Since then, several more groups have elaborated the methods, storing ever larger amounts of data. For example, Erlich and Zielinski~\cite{EZ17} stored 2.11MB of data in DNA, Blawat et al.~\cite{BGHCTIPC16} recovered a data archive of 22MB, and Organick et al.~\cite{Oetal17} stored 200MB of digital information. Yazdi et al.~\cite{YTMZM15} developed a method that offers random read access and rewritable storage using constrained codes. On the other hand, coding theoretic aspects of DNA storage systems have received significant attention recently. The work of \cite{KPM16} discusses error-correcting codes for the DNA sequencing channel, where a possibly erroneous collection of substrings of the original strand is obtained. Codes over unordered sets of sequences, where sequences are affected by a certain number of point errors, such as insertions, deletions and substitutions, have been discussed in \cite{LSWY18,SRB18,SC2018}. In \cite{SRB18}, codes and bounds for a given number of substitutions have been proposed, which require a redundancy that is both logarithmic in the number of sequences and the length of the sequences. Based on a slight adaptation of the model in \cite{LSWY18}, the sequence-subset distance has been introduced and analyzed in \cite{SC2018} and Singleton-like and Plotkin-like code size upper bounds have been derived. In contrast, \cite{KT18} proposes codes for errors that affect whole strands in a storage system that stores multiset of sequences. Recently, codes that can be equipped as primer addresses have been proposed in \cite{YKGM18,CKW19}. A comprehensive survey for DNA-based storage can be found in \cite{YKGMZM15}.

From an information theoretic point of view DNA is fundamentally different than other storage media due to the fact that all information about the ordering of the DNA strands is lost during synthesis. One efficient and practical way to circumvent this limitation is to prepend an index to each strand that denotes the position of the strand in the archive. However, due to errors during synthesis or sequencing, these indices might be erroneous when reading the archive. A naive solution to combat these errors is to protect each index by an error-correcting code. Such an approach however already incurs a redundancy that grows linearly with the number of strands, which is suboptimal, especially for the practically important case, when not all sequences contain errors. In this paper we will analyze the approach of indexing sequences in the presence of errors inside the strands. We propose constructions that efficiently cope with these errors and only have a redundancy that is logarithmic in both the number and length of sequences. Note that the employment of indices is not a necessity and the more general setup of storing an arbitrary set of sequences has been analyzed in \cite{LSWY18}. However, the discussion of indexed-based schemes is practically important due to its simplicity. In this work we study only substitution errors, while insertion and deletion errors are deferred for future work. Also, we present our results for the binary case, while the extension to non-binary alphabets is straightforward.
\section{Channel Model} \label{sec:channel:model}
In this work we study a system where user data is stored in an \emph{indexed set} $\mathcal{S} = \{\ve{x}_1,\dots, \ve{x}_M \}$ of $M$ unordered vectors $\ve{x}_i  \in \Sigma_2^L$, where $i\in[M]\triangleq \{1,2,\dots,M\}$ and $\Sigma_2 = \{0,1\}$. The vectors are also called \emph{sequences} or \emph{strands} in reference to the DNA-based storage system. Hereby, each vector $\ve{x}_i$ has the same length $L$. Throughout the paper, we use that $M = 2^{\beta L}$ for some $0<\beta<1$ such that $\beta L \in \mathbb{N}$ is an integer. Mathematically, an indexed set is characterized as
\begin{equation*}
\mathcal{S} = \{ (\ve{I}(1), \ve{u}_1), (\ve{I}(2), \ve{u}_2), \dots, (\ve{I}(M), \ve{u}_M) \} \subseteq \Sigma_2^L,
\end{equation*}
with sequences $\ve{x}_i = (\ve{I}(i), \ve{u}_i) \in \Sigma_2^L$. Each sequence hereby consists of two parts.  It begins with a prefix $\ve{I}(i) \in \Sigma_2^{\log M}$, also referred to as \emph{index}, of length $ \log M$. This prefix is a unique binary representation of the index $i$ and designates the position of this specific sequence in the data set $\mathcal{S}$. Note that in general it is possible to use any bijective map $\ve{I}(i) : \{1,\dots,M\} \mapsto \Sigma_2^{\log M}$ as index, however in practice this map is usually realized by a standard decimal to binary conversion. The second part of each sequence, $\ve{u}_i \in \Sigma_2^{L-\log M}$, will be referred to as the \emph{data} part of a sequence and can be filled arbitrarily by either user information or redundancy from an error-correcting code, as illustrated later. For convenience, we will abbreviate $L_M\triangleq L-\log M$ throughout the paper. The set of all indexed data sets is
$$
	\mathcal{I}_M^L = \hspace{-.05cm} \left\{  \begin{array}{l}\mathcal{S} =\{(\ve{I}(1), \ve{u}_1), (\ve{I}(2), \ve{u}_2), \dots, (\ve{I}(M), \ve{u}_M) \} : \\
	\ve{u}_i \in \Sigma_2^{L_M} \; \forall\, i=1,\dots,M
	\end{array} \hspace{-.05cm} \right\},
$$
and their total number is $|\mathcal{I}_M^L| = 2^{ML_M}$. Therefore, $\mathcal{I}_M^L$ denotes all feasible channel inputs of the channel, when using indexed sets. The stored set can be corrupted by substitution errors, caused by, e.g., synthesis or sequencing errors and we model the errors by a channel that takes an indexed set $\mathcal{S} \in \mathcal{I}_M^L$ as input and outputs an erroneous outcome of this set based on the following procedure as visualized in Fig.~\ref{fig:channel}. When an indexed data set $\mathcal{S} = \{ \ve{x}_1,\dots,\ve{x}_M \}\in \mathcal{I}_M^L$ has been stored, $M-t$ strands are read correctly and $t$ strands are read in error. These sequences result from clustering and reconstructing a large number of sequences, which has been illustrated and discussed in \cite{Oetal17,LSWY18}. Denote by $\mathcal{F} = \{f_1,f_2,\dots,f_t\} \subseteq [M]$ with $1\leq f_1 < f_2 < \dots<f_t\leq M$ the ordered indices of the sequences that are received in error and $\ve{e}_{f_1}, \dots, \ve{e}_{f_t} \in \Sigma_2^L$ the corresponding error patterns. The index $\ve{I}(i)$ of each erroneous sequence $\ve{x}_i$, $i \in \mathcal{F}$ is affected by at most $\epsilon_1$ substitution errors and the data part $\ve{u}_i$ is affected by at most $\epsilon_2$ substitutions. Therefore, each error vector is composed of two parts $\ve{e}_{f_i} = (\ve{e}^I_{f_i},\ve{e}^D_{f_i})$ of lengths $\log M$ and $L_M$, with Hamming weights $\mathrm{wt}(\ve{e}^I_{f_i}) \leq \epsilon_1$ and $\mathrm{wt}(\ve{e}^D_{f_i}) \leq \epsilon_2$ for all $i\in [t]$. The received set $\mathcal{S}' \subseteq \Sigma_2^L$ can then be written as
$$ \mathcal{S}' = \bigcup_{i=1}^M \left\{ \begin{array}{ll}
\ve{x}_i, & \text{if } i \notin \mathcal{F}, \\
\ve{x}_i + \ve{e}_i, & \text{if }i \in \mathcal{F}
\end{array} \right. . $$
Throughout the paper the $(t,\epsilon_1,\epsilon_2)$-channel will refer to the entity which, given an input set $\mathcal{S} \in \mathcal{I}_M^L$, outputs a received set $\mathcal{S}'$ resulting from arbitrary $\mathcal{F}$ and $\ve{e}_{f_1}, \dots, \ve{e}_{f_t}$ as described above. This set of all possible channel outputs is denoted by $B(\mathcal{S})$.  Note that when there are errors in the indices, the erroneous sequences $\ve{x}_{f_j}' \triangleq \ve{x}_{f_j} + \ve{e}_{f_j}$, $j \in [t]$ are not necessarily distinct from each other or from the error-free sequences and in this case these sequences adjoin and appear as a single sequence at the receiver. Therefore the number of received sequences can be less than $M$, i.e., $M-t \leq |\mathcal{S}'| \leq M$. In particular here it is also possible that the received set $\mathcal{S}' \notin \mathcal{I}_M^L$, since  some indices might not be present in the received set or others might appear multiple times.
Another particularity of the channel is that different error patterns $\mathcal{F}$ and $\ve{e}_{f_1},\dots,\ve{e}_{f_t}$ might lead to the same channel output $\mathcal{S}'$. We will use the following standard definition of an error-correcting code.
\begin{defn}[$(t,\epsilon_1,\epsilon_2)$-indexed-set code] {\label{def:ecc}}
	A code $\mathcal{C} \subseteq \mathcal{I}_M^L$ is called a $(t,\epsilon_1,\epsilon_2)$-indexed-set code, if $B(\mathcal{S}_1) \cap B(\mathcal{S}_2) = \emptyset$ for every pair $\mathcal{S}_1,\mathcal{S}_2 \in \mathcal{C}$ with $\mathcal{S}_1\neq \mathcal{S}_2$. Accordingly, the redundancy of an indexed-set code $\mathcal{C} \subseteq \mathcal{I}_M^L$ is defined to be
	$$ r(\mathcal{C}) = ML_M - \log |\mathcal{C}|. $$
\end{defn}%
By this definition, an indexed-set code is a set of codewords for which, for each channel output $\mathcal{S}' \subseteq \Sigma_2^L$, there exists at most one codeword which could have resulted in this exact channel output $\mathcal{S}'$. Note that here, each codeword is not a vector, as in the standard channel coding problem, but a set of indexed vectors. In this paper, we distinguish between errors in the index of sequences and data part of the sequences due to the following reasons. It is observed that the sequencing error rates at the beginning of DNA strands are lower with several sequencing technologies \cite{EZ17,HMG18,Oetal17}. Second, from a theoretical point of view, errors inside the indices have a different character than those in the data part, as they do not affect data directly but hinder the correct identification of the strand order. We will also elaborate in this paper that the redundancy required to correct errors in the indices is significantly smaller than that in the data part of sequences. Finally, the channel model is strongly connected to the more general model presented in \cite{LSWY18} as follows. 
\begin{enumerate}
	\item Each $(0,t,\epsilon)_\sub$-correcting code \cite{LSWY18} is a $(t,\epsilon_1,\epsilon_2)$-indexed-set code, if $\epsilon_1+\epsilon_2\leq \epsilon$.
	\item Each $(t,\epsilon_1,\epsilon_2)$-indexed-set code is a $(0,t,\epsilon)_\sub$-correcting code \cite{LSWY18}, if $\epsilon \leq \min(\epsilon_1,\epsilon_2)$.
\end{enumerate}
\section{Construction}
Finding codes that can correct errors from the DNA-storage channel, one faces two main challenges that have to be tackled. To begin with, substitution errors that are solely in the data part of the sequences can be corrected by standard error-correcting schemes, such as tensor-product codes \cite{Wol06}, which we will discuss in more detail later. However, errors in the indices of sequences will corrupt the ordering of the sequences, which hinders the direct employment of tensor-product codes. We therefore will construct a code that first enables to reconstruct the correct ordering of the sequences using so called \emph{anchors}, and then uses a tensor-product code to correct the errors in the data part of the sequences. The anchors are defined as follows.
\begin{defn}[Anchor]
	Let $l,t,\epsilon_1,\epsilon_2 \in \mathbb{N}$ and $\ve{a}_1,\dots, \ve{a}_M \in \Sigma_2^l$ be $M$ vectors of length $l$ with $2^l \geq M$. Further, denote by $\mathsf{MDS}[M,2t]$ a maximum-distance-separable (MDS) code of length $M$ and redundancy $2t$ over the field $\Sigma_{2^l}$. The set of \emph{anchor} vectors $\mathcal{A}(l,t,\epsilon_1,\epsilon_2)$ is defined to be
	$$ \mathcal{A}(l,t,\epsilon_1,\epsilon_2) \hspace{-.08cm}=\hspace{-.08cm} \left\{\hspace{-.16cm} \begin{array}{l}
	(\ve{a}_1,\dots,\ve{a}_M) \in \Sigma_2^{Ml}\hspace{-.08cm} : \; \forall i,j\hspace{-.08cm} \in\hspace{-.08cm} [M], i\neq j: \\
	d(\ve{a}_i, \ve{a}_j) > 2\epsilon_2, \text{ if } d(\ve{I}(i), \ve{I}(j)) \leq 2\epsilon_1, \\
	(\ve{a}_1,\dots,\ve{a}_M) \in \mathsf{MDS}[M,2t]  
	\end{array} \hspace{-.16cm} \right\}. $$
	That is, if the indices $\ve{I}(i), \ve{I}(j)$ of two vectors $\ve{a}_i, \ve{a}_j$ have distance at most $2\epsilon_1$, the vectors have distance more than $2\epsilon_2$. Further, the equivalents of the vectors $\ve{a}_1,\dots,\ve{a}_M$ in $\Sigma_{2^l}$ are a codeword of an MDS code with minimum distance $2t+1$.
\end{defn}
This definition implies that the anchor vectors have both a large intra-anchor distance between vectors of one anchor and a large inter-anchor distance between two anchors due to the MDS code. Note that for $2\epsilon_1 = \log M$ and $t=0$ this definition is equivalent to a standard error-correcting code, which corrects $\epsilon_2$ errors. The redundancy required to force such a constraint on a collection of vectors will be calculated later. For the case of $t=0$, the set $\mathcal{A}(l,0,\epsilon_1,\epsilon_2)$ is called clustering-correcting code, and explicit constructions which require only one bit of redundancy and can be encoded and decoded efficiently can be found in \cite{SYLW19}. The anchoring property will be used to reconstruct the ordering of the sequences. After the ordering of sequences is restored, it is possible to correct the errors in the sequences using tensor-product codes \cite{Wol06}, which are defined as follows.
\begin{defn}[Tensor-product code]
	Let $\mathcal{C}_1 \subseteq \Sigma_{2}$ be a linear $[L_M,L_M-r_1,\epsilon_2]$ binary $\epsilon_2$-error-correcting code of length $L_M$, redundancy $r_1$ and parity-check matrix \mbox{$\ve{H}_1 \in \Sigma_2^{r_1\times L_M}$} and let $\mathcal{C}_2 \subseteq \Sigma_{2^{r_1}}$ be a linear $[M,M-r_2,t]$ code over the field $\Sigma_{2^{r_1}}$. The \emph{tensor-product code} is then defined to be
	$$
		\mathsf{TPC}(t,\epsilon_2) = \left\{ \begin{array}{ll}
		(\ve{u}_1,\dots,\ve{u}_{M}) \in \Sigma_{2}^{ML_M}:  \\
		(\ve{s}_1,\dots,\ve{s}_{M}) \in \mathcal{C}_2
		\end{array} \right\},
	$$
	where $\ve{s}_i = \ve{u}_i \ve{H}_1^\mathrm{T}$ are syndromes whose equivalents in the finite field $\Sigma_{2^{r_1}}$ form a codeword of $\mathcal{C}_2$. The overall redundancy of the tensor-product code is $r_1r_2$ bits.
\end{defn}
Correcting errors using the tensor-product code is done as follows \cite{Wol06}. Assume the word $\ve{U} = (\ve{u}_1',\dots,\ve{u}_{M}')$ is received, where at most $t$ vectors $\ve{u}_i'$ have been affected by at most $\epsilon_2$ errors each. The receiver first computes the syndromes $\ve{s}_i' = \ve{u}_i' \ve{H}_1^\mathrm{T}$ of all vectors. Since there are at most $t$ syndromes corrupted, the correct syndromes $\ve{s}_i$ can be recovered using the code $\mathcal{C}_2$. Now, in each row, $\epsilon_2$ errors can be corrected using the knowledge of the correct syndrome $\ve{s}_i$ and the code $\mathcal{C}_1$. Combining the anchoring property with the tensor-product code leads to the following construction.
\begin{construction} \label{con:anchor}
	Let $l,t,\epsilon_1,\epsilon_2 \in \mathbb{N}$ with $l \geq \log M
	$. Further, $\mathsf{TPC}(t,\epsilon_2)$ denotes a tensor-product code over an array of size $M\times L_M$. We define the construction $\mathcal{C}_\mathrm{A} \subseteq \mathcal{I}_M^L$ as
	$$
	\mathcal{C}_\mathrm{A} =  \left\{  \begin{array}{l}\mathcal{S} =\{(\ve{I}(1), \ve{a}_1, \ve{v}_1), \dots, (\ve{I}(M), \ve{a}_M,\ve{v}_M) \} : \\
	(\ve{a}_1,\dots,\ve{a}_M) \in \mathcal{A}(l,t,\epsilon_1,\epsilon_2), \\
	((\ve{a}_1,\ve{v}_1), \dots, (\ve{a}_M, \ve{v}_M)) \in \mathsf{TPC}(t,\epsilon_2)
	\end{array} \right\}.
	$$
\end{construction}
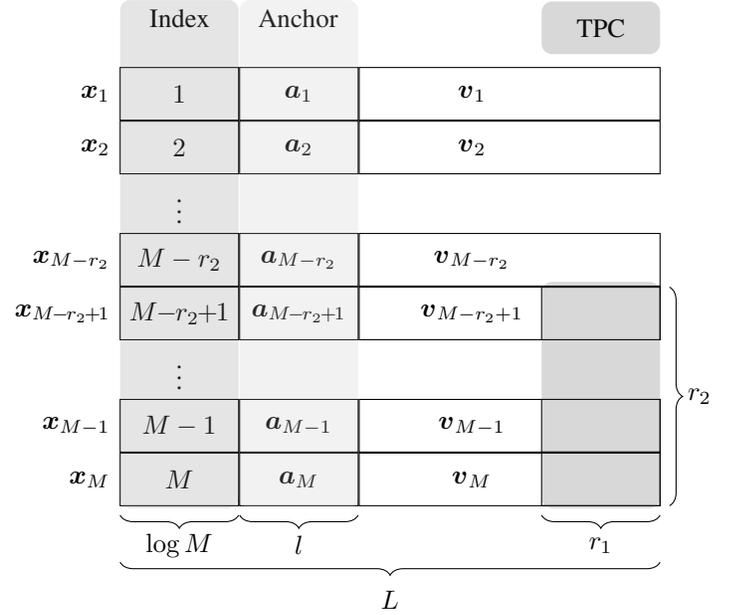
\begin{figure}
	\newcommand{\xwidth}{0cm}
	\newcommand{\indexwidth}{1.575cm}
	\newcommand{\ancorwidth}{1.575cm}
	\newcommand{\datawidth}{4cm}
	\newcommand{\datashift}{-1cm}
	\newcommand{\mht}{0.7cm}
	
	\tikzstyle{block2} = [rectangle, draw, minimum height=\mht]
	\hspace{-.25cm}\begin{tikzpicture}
	
	\node[block2 ,minimum width=\indexwidth] (index) {$1$};
	\node[left= 0pt of index,minimum width=\xwidth] {$\ve{x}_1$};
	\node[block2, right= 0pt of index,minimum width=\ancorwidth] (information) {$\ve{a}_1$};
	\node[block2, right= 0pt of information, minimum width=\datawidth] {\hspace{\datashift}$\ve{v}_1$};
	
	\node[block2, below= 0pt of index,minimum width=\indexwidth] (index2) {$2$};
	\node[left= 0pt of index2,minimum width=\xwidth] {$\ve{x}_2$};
	\node[block2, right= 0pt of index2, minimum width=\ancorwidth] (information2) {$\ve{a}_2$};
	\node[block2, right= 0pt of information2, minimum width=\datawidth] {\hspace{\datashift}$\ve{v}_2$};

	\node[below = 0pt of index2,minimum height=\mht] (dots) {$\vdots$};

	\node[block2, below= 0pt of dots,minimum width=\indexwidth] (indexMR) {$M-r_2$};
	\node[left= 0pt of indexMR,minimum width=\xwidth] {$\ve{x}_{M-r_2}$};
	\node[block2, right= 0pt of indexMR,minimum width=\ancorwidth] (informationMR) {$\ve{a}_{M-r_2}$};
	\node[block2, right= 0pt of informationMR, minimum width=\datawidth] {\hspace{\datashift}$\ve{v}_{M-r_2}$};
	
	\node[block2, below= 0pt of indexMR,minimum width=\indexwidth] (indexMR1) {$M\hspace{-.1cm}-\hspace{-.1cm}r_2\hspace{-.1cm}+\hspace{-.1cm}1$};
	\node[left= 0pt of indexMR1,minimum width=\xwidth] {$\ve{x}_{M\hspace{-.03cm}-\hspace{-.03cm}r_2\hspace{-.03cm}+\hspace{-.03cm}1}$};
	\node[block2, right= 0pt of indexMR1,minimum width=\ancorwidth] (informationMR1) {$\ve{a}_{M\hspace{-.03cm}-\hspace{-.03cm}r_2\hspace{-.03cm}+\hspace{-.03cm}1}$};
	\node[block2, right= 0pt of informationMR1, minimum width=\datawidth] (dataMR1) {\hspace{\datashift}$\ve{v}_{M-r_2+1}$};
	\node[block2, right= 0pt of dataMR1,minimum width=\indexwidth, anchor=east] (tpcMR1) {};

	\node[below = 0pt of indexMR1,minimum height=\mht] (dots2) {$\vdots$};

	\node[block2, below= 0pt of dots2,minimum width=\indexwidth] (indexM1) {$M-1$};
	\node[left= 0pt of indexM1,minimum width=\xwidth] {$\ve{x}_{M-1}$};
	\node[block2, right= 0pt of indexM1,minimum width=\ancorwidth] (informationM1) {$\ve{a}_{M-1}$};
	\node[block2, right= 0pt of informationM1, minimum width=\datawidth] (dataM1) {\hspace{\datashift}$\ve{v}_{M-1}$};
	\node[block2, right= 0pt of dataM1,minimum width=\indexwidth, anchor=east] {};
	
	\node[block2, below= 0pt of indexM1,minimum width=\indexwidth] (indexM) {$M$};
	\node[left= 0pt of indexM,minimum width=\xwidth] {$\ve{x}_M$};
	\node[block2, right= 0pt of indexM,minimum width=\ancorwidth] (informationM) {$\ve{a}_M$};
	\node[block2, right= 0pt of informationM, minimum width=\datawidth] (dataM) {\hspace{\datashift}$\ve{v}_M$};
	\node[block2, right= 0pt of dataM,minimum width=\indexwidth, anchor=east] (tpcM) {};
	
	\node at (0,1) {Index};
	\node at (\indexwidth/2+\ancorwidth/2,1) {Anchor};
	\node at ($(tpcMR1) + (0, 4*\mht+1cm)$) (TPClabel) {TPC};
	
	\draw [decorate,decoration={brace,amplitude=5pt, mirror}] ($(indexM) + (-\indexwidth/2,-\mht/2-0.125cm)$) -- ($(indexM) +  (\indexwidth/2,-\mht/2-0.125cm)$) node [black,midway,yshift=-0.4cm] { $\log M$};
	
	\draw [decorate,decoration={brace,amplitude=5pt, mirror}] ($(informationM) + (-\ancorwidth/2,-\mht/2-0.125cm)$) -- ($(informationM) +  (\ancorwidth/2,-\mht/2-0.125cm)$) node [black,midway,yshift=-0.4cm] {$l$};
	
	\draw [decorate,decoration={brace,amplitude=5pt, mirror}] ($(tpcM) + (-\indexwidth/2,-\mht/2-0.125cm)$) -- ($(tpcM) +  (\indexwidth/2,-\mht/2-0.125cm)$) node [black,midway,yshift=-0.4cm] {$r_1$};
	
	\draw [decorate,decoration={brace,amplitude=5pt}] ($(tpcMR1) + (\indexwidth/2+0.125cm,\mht/2)$) -- ($(tpcM) + (\indexwidth/2+0.125cm,-\mht/2)$) node [black,midway,xshift=0.4cm] {$r_2$};
	
	\draw [decorate,decoration={brace,amplitude=5pt, mirror}] ($(indexM) + (-\indexwidth/2,-\mht/2-.75cm)$) -- ($(dataM) +  (\datawidth/2,-\mht/2-.75cm)$) node [black,midway,yshift=-0.5cm] {$L$};

	\fill[gray, opacity=0.2, rounded corners] (-\indexwidth/2, 1.25cm) rectangle ($(indexM) + (0.5*\indexwidth, -0.55*\mht)$);
	
	\fill[lightgray, opacity=0.2, rounded corners] ($(information) + (-\ancorwidth/2, 1.25cm)$) rectangle ($(informationM) + (0.5*\ancorwidth, -0.55*\mht)$);
	
	\fill[darkgray, opacity=0.2, rounded corners] ($(tpcMR1) + (-\indexwidth/2, 0.6*\mht)$) rectangle ($(tpcM) + (\indexwidth/2, -0.55*\mht)$);
	
	\fill[darkgray, opacity=0.2, rounded corners] ($(TPClabel) + (-\indexwidth/2, 0.5*\mht)$) rectangle ($(TPClabel) + (\indexwidth/2, -0.5*\mht)$);
	\end{tikzpicture}
	\vspace{-.75cm}
	\caption{Schematic of Construction \ref{con:anchor}}
	\vspace{-.5cm}
\end{figure}
Note that with this construction, the anchors $\ve{a}_1,\dots,\ve{a}_M$ can also contain user data. The correctness of Construction \ref{con:anchor} and its decoding algorithm are presented in the following.
\begin{lemma}
	Construction \ref{con:anchor} is a $(t,\epsilon_1,\epsilon_2)$-indexed-set code.
\end{lemma}
\begin{proof}
	We will prove the correctness of Construction \ref{con:anchor} by providing an algorithm that can be used to correct errors from the $(t,\epsilon_1,\epsilon_2)$-channel. The decoding algorithm can be split into the following two steps.
	\begin{enumerate}
		\item Retrieve the correct order of sequences using the anchoring property of $\ve{a}_1,\dots,\ve{a}_M$.
		\item Correct errors inside the sequences using the tensor-product code $\mathsf{TPC}(t,\epsilon_2)$.
	\end{enumerate}
	Assume $\mathcal{S} = \{ \ve{x}_1,\dots,\ve{x}_M \} \in \mathcal{C}_\mathrm{A}$ has been stored and $\mathcal{S}'= \{ \ve{x}_1',\dots,\ve{x}_M' \} \in B(\mathcal{S})$ has been received after transmission over a $(t,\epsilon_1,\epsilon_2)$-channel. We will write $\ve{x}_i' = (\ve{I}(i'),\ve{a}_i',\ve{v}_i')$, which is either $\ve{x}_i' = \ve{x}_i$, if the sequence was received correctly, i.e., $i \notin \mathcal{F}$, or $\ve{x}_i' = \ve{x}_i + \ve{e}_i$, if the sequence was received in error, i.e., $i \in \mathcal{F}$. This correct ordering of received sequences is however only used to simplify notation and is not known to the receiver, as the indices $\ve{I}(i')$ can be erroneous. Note that due to the anchoring property, it is guaranteed that an erroneous sequence can never adjoin with another sequence and therefore $|\mathcal{S}'| = M$. The anchors can be fully recovered using their MDS property as follows. Declare all positions $i\in [M]$, where there is not exactly one index present, i.e., $i : |\{j : \ve{I}(j') = \ve{I}(i) \}| \neq 1$ as erasures, and fill all remaining positions with the corresponding anchors $\ve{a}_{i}'$. Although some anchors might have the wrong position, decoding the resulting vector of length $M$ with a unique decoding algorithm yields the correct anchors $\ve{a}_1,\dots,\ve{a}_M$ (cf. \cite[Con. 1]{LSWY18}). Using the anchors, it is possible to assign each sequence $\ve{x}_j'$ to its correct position $i$ by finding the single sequence $\ve{x}_j' \in \mathcal{S}'$ with $d(\ve{I}(i),\ve{I}(j')) \leq \epsilon_1$ and $d(\ve{a}_i,\ve{a}_j') \leq \epsilon_2$. There is exactly one sequence $j=i$ with that property. Assume on the contrary, there is more than one sequence (apart from the correct sequence $\ve{x}_i'$), which fulfills this property. Then, there would be a sequence $\ve{x}_j'$, $j\neq i$ with $d(\ve{I}(i),\ve{I}(j')) \leq \epsilon_1$ and $d(\ve{a}_i,\ve{a}_j') \leq \epsilon_2$, which implies that $d(\ve{I}(i),\ve{I}(j)) \leq 2 \epsilon_1$ and also $d(\ve{a}_i,\ve{a}_j) \leq 2\epsilon_2$, which contradicts the anchoring property. We therefore can reconstruct the array $((\ve{a}_1',\ve{v}_1'), \dots, (\ve{a}_M', \ve{v}_M'))$ in the correct order. Since each row $(\ve{a}_1',\ve{v}_1')$ has at most $\epsilon_2$ errors, these errors can be corrected using the tensor-product code, which completes the proof of the correctness of Construction \ref{con:anchor}.
\end{proof}
The redundancy of Construction \ref{con:anchor} can be decomposed into the redundancy required for the anchoring property and the redundancy of the tensor-product code and is given as follows.
\begin{thm}
	For any $t,\epsilon_1,\epsilon_2$ the redundancy of $\mathcal{C}_\mathrm{A}$ is
	$$ r(\mathcal{C}_\mathrm{A}) = r_\mathrm{A} + r_1r_2, $$
	where $r_\mathrm{A} = 2tl - M \log (1 - 2^{-l}B_{2\epsilon_1}(\log M) B_{2\epsilon_2}(l))$. Therefore, for fixed $t,\epsilon_1,\epsilon_2$, and arbitrary small $\delta>0$, for $M\rightarrow \infty$ there exists an explicit construction $\mathcal{C}_\mathrm{A}$ with redundancy
	$$ r(\mathcal{C}_A) \leq (4t+2\delta) \log M + 2t \epsilon_2 \lceil\log L_M\rceil+ 1+ o(1). $$
\end{thm}
\begin{proof}
	From the cardinality of clustering-correcting codes \cite{SYLW19} and the fact that the MDS code with redundancy $2t$ has $2^{2tl}$ cosets, there exists one coset of the MDS code with
	$$ |\mathcal{A}(l,t,\epsilon_1,\epsilon_2)| \geq \frac{1}{2^{2tl}} (2^l -B_{2\epsilon_1}(\log M) B_{2\epsilon_2}(l))^{M} $$
	by the pigeonhole principle. From this follows the redundancy $r_\mathrm{A}$ required for the anchoring property. Next, the redundancy of the tensor-product codes is $r_1r_2$. Using alternant codes \cite[ch. 5]{Rot06} $\mathcal{C}_1$ and $\mathcal{C}_2$, we obtain redundancies $r_1=\epsilon_2 \lceil \log L_M\rceil$ and $r_2 = 2t \lceil \frac{\log M}{r_1} \rceil$, if $r_1 \leq \log M$ and $r_2 = 2t$, otherwise. Using $l=(1+\delta) \log M$ yields $r_\mathrm{A} = 2t(1+\delta) \log M + o(1)$ and the asymptotic bound follows.
\end{proof}
Note that for $t=1$, the construction can be improved by using a Hamming code for $\mathcal{C}_2$ and an $\mathsf{MDS}[M,1]$ code with redundancy $1$ for the anchors is sufficient, which yields a redundancy of approximately $2\log M + \epsilon_2\log L_M + o(1)$.
\section{Sphere Packing Bound}
The derivation of the sphere packing bound is based on the sets $B(\mathcal{S})$ of possible outputs of the channel, when $\mathcal{S} \in \mathcal{I}_M^L$ is the input. The bound is derived by using the fact that $B(\mathcal{S})$ must be distinct for different codewords $\mathcal{S}$ to guarantee unique decoding to one codeword. In this and the following section, we will abbreviate the size of the Hamming ball of radius $r$ by $B_r(n)\triangleq \sum_{i=0}^{n}\binom{n}{i}$. The main result is as follows.
\begin{thm}
	The cardinality of any $(t,\epsilon_1,\epsilon_2)$-indexed-set code $\mathcal{C} \subseteq \mathcal{I}_M^L$ is at most
	$$ |\mathcal{C}| \leq \frac{2^{ML_M}}{\binom{M}{t} (B_{\epsilon_2}(L_M)-1)^t}. $$
	Therefore, the redundancy is at least
	$$ r(\mathcal{C}) \geq t\log M  + t \epsilon_2 \log (L_M) - t \log (t \epsilon_2^{\epsilon_2}) . $$
\end{thm}
\begin{proof}
	Let $\mathcal{C} \subseteq \mathcal{I}_M^L$ be a $(t,\epsilon_1,\epsilon_2)$-indexed-set code. We consider first the case that $\epsilon_1=0$, i.e., there are only errors outside the indices and therefore all erroneous outcomes $\mathcal{S}' \in B(\mathcal{S}) \cap \mathcal{I}_M^L$ are again indexed sets. Due to the distinctness of error balls, every code $\mathcal{C} \subseteq \mathcal{I}_M^L$ satisfies $|\mathcal{C}| \cdot \min_{\mathcal{S} \in \mathcal{I}_M^L}|B(\mathcal{S})\cap \mathcal{I}_M^L| \leq |\mathcal{I}_M^L|$. Using this inequality we bound the code size $|\mathcal{C}|$ from above. Specifically, for all $\mathcal{S} \in \mathcal{I}_M^L$, we bound the number of erroneous outcomes $|B(\mathcal{S}) \cap \mathcal{I}_M^L|$ which are again indexed sets from below. Distinct elements $\mathcal{S}' \in B(\mathcal{S}) \cap \mathcal{I}_M^L$ are obtained as follows. For $\epsilon_1=0$ the indices of each sequence can be omitted and the stored set can be viewed as a binary array of $M$ rows and $L_M$ columns, where each row corresponds to one sequence. The number of possible error patterns is therefore
	$$|B(\mathcal{S}) \cap \mathcal{I}_M^L| \geq \binom{M}{t} (B_{\epsilon_2}(L_M)-1)^t,$$
	as there are $\binom{M}{t}$ ways to choose the erroneous rows and $B_{\epsilon_2}(L_M)-1$ possible substitution patterns per row. Finally, the case $\epsilon_1=0$ is a special case of $\epsilon_1>0$, as there are \emph{up to} $\epsilon_1$ errors inside the indices and thus the above bound also holds for arbitrary $\epsilon_1>0$ which concludes the proof.
\end{proof}
Note that by the definition of the channel it is possible that errors occur in the index of a sequence. However considering these errors for the sphere packing bound does not improve the bound, as we will illustrate in the following. Let us for simplicity assume that there has only been one error in the \mbox{$i$-th} sequence, and compare the two cases, where first, the error is in the data part, i.e., $t=\epsilon_2=1$ and $\epsilon_1 = 0$, and second, the error is in the index, i.e.,  $t=\epsilon_1=1$ and $\epsilon_2 = 0$. In the first case, it is sufficient to use a Hamming code of length $ML_M$ and redundancy $\log (ML_M)$, which is able to correct the single substitution, as the receiver can correctly concatenate the received sequences. On the other hand, when the error occurs inside the index of sequence $i$, resulting in index $j$, the receiver will see two sequences with the same index $j$ and no sequence with index $i$. In this case, the receiver only has to decide which of the two sequences with the index $j$ belongs to the position $i$. As this is merely a binary decision, from a sphere packing point of view, a redundancy of roughly a single bit is sufficient to correct this error. This surprisingly  indicates that errors inside indices of sequences are less harmful than those inside the data fields of sequences.

\section{Gilbert-Varshamov Bound}
In the last section we have derived upper bounds on the cardinality of error-correcting codes for indexed-set codes. On the other hand, we will now show how to find lower bounds on the achievable size of such error-correcting indexed-set codes based on Gilbert-Varshamov-like sphere covering arguments. For convenience, in the following we denote by $V(\mathcal{S})$ the set of indexed sets $\tilde{\mathcal{S}} \in \mathcal{I}_M^L$ which have intersecting errors ball with $\mathcal{S}\in \mathcal{I}_M^L$, i.e., $B(\mathcal{S}) \cap B(\tilde{\mathcal{S}}) \neq \emptyset$.
\begin{thm}
	There exists a $(t,\epsilon_1,\epsilon_2)$-indexed-set code $\mathcal{C}\subseteq \mathcal{I}_M^L$ with cardinality at least
	$$ |\mathcal{C}| \geq \frac{2^{ML_M}}{\binom{M}{t}^2 (B_{\epsilon_2}(L_M))^{2t} (t!^2 + \frac{t}{M-t}(B_{\epsilon_1}(\log M))^{2t} ) }. $$
	Therefore, for fixed $t,\epsilon_1,\epsilon_2$ and $M\rightarrow \infty$, there exists a $(t,\epsilon_1,\epsilon_2)$-indexed-set code $\mathcal{C}\subseteq \mathcal{I}_M^L$ with redundancy at most
	$$ r(\mathcal{C}) \leq 2t \log M + 2t \epsilon_2 \log L_M - 2t\log \epsilon_2! + o(1).$$
\end{thm}
\begin{proof}
	Based on an iterative procedure, it can be shown that there exists a $(t,\epsilon_1,\epsilon_2)$-indexed-set code $\mathcal{C} \subseteq \mathcal{I}_M^L$ with $|\mathcal{C}| \cdot \max_{\mathcal{S} \in \mathcal{I}_M^L}|V(\mathcal{S})|  \geq |\mathcal{I}_M^L|$. Bounding $|V(\mathcal{S})|$ from above for all $\mathcal{S}\in\mathcal{I}_M^L$ will be the main task in the following. Let $B_I(\mathcal{S}) \triangleq B(\mathcal{S}) \cap \mathcal{I}_M^L$ be the set of erroneous sets which are indexed sets and $B_N(\mathcal{S}) \triangleq B(\mathcal{S}) \setminus B_I(\mathcal{S})$. Further distinguish between $V_I(\mathcal{S}) \triangleq \{ \tilde{\mathcal{S}}\in \mathcal{I}_M^L: B_I(\mathcal{S}) \cap B(\tilde{\mathcal{S}}) \neq \emptyset \}$ and $V_N(\mathcal{S}) \triangleq V(\mathcal{S}) \setminus V_I(\mathcal{S})$ and note that $V(\mathcal{S}) = V_I(\mathcal{S}) \cup V_N(\mathcal{S})$. We first count $|V_I(\mathcal{S})|$. To begin with, $|B_I(\mathcal{S})|  \leq  \binom{M}{t}t! (B_{\epsilon_2}(L_M))^t$, as there are $\binom{M}{t}$ ways to choose the erroneous sequences $\mathcal{F}$. For one fixed $\mathcal{F}$, there are at most $t!$ error patterns for the errors in the indices $\ve{e}_{f_1}^{(1)}, \dots, \ve{e}_{f_t}^{(1)}$ that yield indexed sets, as any permutation of erroneous sequences is potentially possible. For each such choice there are again at most $(B_{\epsilon_2}(L_M))^t$ ways to distribute the errors in the data fields of the $t$ erroneous sequences. From each $\mathcal{S}' \in B_I(\mathcal{S})$, there are again at most $|B_I(\mathcal{S}')|$ ways to arrive at a valid set $\tilde{\mathcal{S}} \in \mathcal{I}_M^L$ and thus $|V_I(\mathcal{S})| \leq |B_I(\mathcal{S})|^2$. Next we count $|V_N(\mathcal{S})|$. The number of elements in the error ball is at most $|B_N(\mathcal{S})| \leq \binom{M}{t} B_{\epsilon_1}(\log M) B_{\epsilon_2}(L_M)$, as this is the maximum number of error patterns. Let $\mathcal{S}' \in B_N(\mathcal{S})$ and denote by $t_N(\mathcal{S}')$ the number of indices that are not present in $\mathcal{S}'$. Then the number of sets $\tilde{\mathcal{S}} \in \mathcal{I}_M^L$ with $\mathcal{S}' \in B(\tilde{\mathcal{S}})$ is at most $(B_{\epsilon_1}(\log M))^t  \binom{M}{t-t_N(\mathcal{S}')} (B_{\epsilon_2}(L_M))^t$, as $t_N(\mathcal{S}')$ sequences have to be distorted in a way such that their indices match the missing indices. And thus, there are only $B_{\epsilon_1}(\log M)$ options per missing index in $\mathcal{S}'$. The remaining erroneous sequences can be chosen arbitrarily. Using $t_N(\mathcal{S}') \geq 1$ for all $\mathcal{S}' \in B_N(\mathcal{S})$ yields the theorem.
\end{proof}
\section{Conclusion}
\vspace{-.05cm}
In this paper, we have discussed codes, where each codeword is an indexed set of several vectors. The proposed construction significantly improves the redundancy $2tL$ from \cite[Con. 1]{LSWY18} to $4t\log M + 2t \epsilon_2 \log L_M$, and approaches the sphere-packing bound $t\log M + t\epsilon_2 \log L_M$ up to a factor of $4$ and a factor of $2$ for $t=1$. Further, our results surprisingly indicate that errors within the index of sequences seem to be less harmful than errors in the data part of sequences. This is in sharp contrast to current technologies that often rely on extra codes, which only protect the index of sequences in order to guarantee correct ordering of sequences.

\vspace{-.08cm}

\bibliography{IEEEabrv,ref}
\bibliographystyle{IEEEtranS}

\end{document}